\documentclass[conference]{IEEEtran}
\usepackage{cancel}
\usepackage{epsfig,epstopdf,graphicx,subfigure,psfrag,amsmath,cases}
\usepackage{latexsym,amssymb,amsmath,epsfig,subfigure,algorithm,amsthm}
\usepackage{nopageno}
\usepackage{algorithmic}
\usepackage{color}
\usepackage{url}
\usepackage{scrtime}
\usepackage{bm}
\usepackage{graphicx}
\usepackage{bbding}
\usepackage{multicol}
\usepackage{graphicx}
\usepackage{amsthm}
\usepackage{cite}
\usepackage{array}
\usepackage[utf8]{inputenc}
\usepackage[english]{babel}
 \usepackage{setspace}
 \usepackage{subfloat}
\addto\captionsenglish{}

\title{Secure Communication in Multifunctional Active Intelligent Reflection Surface-assisted Systems}
\author{\IEEEauthorblockN{Shaokang Hu, Derrick Wing Kwan Ng, and Jinhong Yuan \\
          School of Electrical Engineering and Telecommunications, University of New South Wales, Sydney, Australia\vspace{-8mm}}}

{}
\newtheorem{theorem}{Theorem}

\newtheorem{Proposition}{Proposition}

\begin{document}
\maketitle

\thispagestyle{empty}

\begin{abstract}
In this paper, the secure performance of multiuser multiple-input single-output wireless communications systems assisted by a multifunctional active intelligent reflection surface (IRS) is investigated. The active IRS can simultaneously reflect and amplify the incident signals and emit artificial noise to combat potential wiretapping. We minimize the total system power consumption by designing the phase, amplitude, and IRS mode selection of the active IRS elements, as well as the precoder and artificial noise vector of the base station (BS). The design is formulated as a non-convex optimization problem guaranteeing communication security.  To tackle the problem, this paper proposes an iterative alternating algorithm
to obtain an effective sub-optimal solution. The simulation results show that the proposed scheme offers superior secure performance over all the considered baseline schemes, especially when the number of eavesdropper antennas is more than that of the BS.
\end{abstract}
\large\normalsize
\section{Introduction}
With the rapid development of meta-materials and electromagnetic materials, passive intelligent reflecting surfaces (IRSs) have been advocated as a promising solution to enable future wireless communications as they provide a new degrees-of-freedom (DoF) to intelligently shape the propagation condition of wireless channels \cite{wu2019intelligent}.
However, the performance of passive IRS-assisted systems is constrained by the inherent ``double-fading'' attenuation. In particular, since the received signals propagate through the cascaded reflection link, they suffer severe attenuation from large-scale fading twice\cite{8936989}. To circumvent this challenge, active IRS has been proposed in \cite{zhang2021active,xu2021resource,9377648}.  Different from passive IRSs, active reflection-type amplifiers are integrated into each element of active IRSs that can be realized by off-the-shelf components such as current-inverting converters\cite{lonvcar2019ultrathin}. As a result, an active IRS can further amplify the incident signals with controllable amplitudes and phases to mitigate the impacts of large-scale fading. Nevertheless, one major drawback of active IRS is that its associated noise power, e.g. thermal noise, scales with the transmit power of the active IRS that may result in a system performance bottleneck\cite{zhang2021active}.

On the other hand, although IRSs can improve the received signal strength, it also increases the susceptibility to potential eavesdropping.
As a result, in IRS-assisted wireless communications, communication security is a critical concern \cite{9483903}. To realize secure communication, numerous works have adopted the physical layer security techniques, e.g. artificial noise (AN)-aided beamforming\cite{sun2018robust}. Conventionally, a properly designed AN signal is inserted into the transmitted signals by the base station (BS) deliberately to confuse the signal detection at potential eavesdroppers. The rationale behind the AN design is to effectively reduce the interference leakage to the legitimate users while degrading the intercepted signal quality at the potential eavesdroppers via optimizing the direction of AN.
As such, the BS is able to improve the secrecy performance by adaptively adjusting the spatial directions of AN and the information signals jointly via spatial beamforming\cite{sun2018robust}. Unfortunately, this technique can only guarantee communication security when the total number of antennas equipped at the transmitter is larger than the receive ones at the eavesdropper. For this reason, most of the existing works, e.g. \cite{9483903,zhou2021secure}, ideally assume that the number of eavesdropper antennas is less than that of the BS even though it may not hold in practice.  To overcome the aforementioned difficulties, this paper designs secure communication systems assisted by a multifunctional active IRS with the ability to emit AN to combat potential wiretapping. In particular, the originally detrimental thermal noise at the active IRS is exploited as beneficial AN for secrecy communication provisioning. Thus, effective secure communications can be obtained even when the number of eavesdroppers' antennas has a competitive edge over the BS.

In this paper,  we formulate the resource allocation algorithm design for active IRS-assisted multiuser multiple-input single-output (MISO) wireless systems with the presence of a multi-antenna eavesdropper to minimize the total system power consumption. In contrast to previously published works \cite{zhang2021active,xu2021resource,9377648}, this paper is the first attempt to propose a novel versatile active IRS to realize secure wireless communications. In particular, this work allows each active IRS element operating in either the reflection mode or the jamming mode to swing the balance between improving the signal quality of the users and communication security. By smartly optimizing the IRS  mode selection, phase and amplitude of active IRS elements, and the precoder and AN vector of the BS, the active IRS can adaptively serve as a separate jamming source against potential eavesdropping, while reflecting the impinging signals to the desired legitimate users.
To handle the challenging non-convex optimization problem, we propose an alternating optimization (AO)-based algorithm to obtain an effective sub-optimal solution. Simulation results show that the proposed scheme outperforms all the baseline schemes, especially when the number of eavesdroppers' antenna is larger than that of the BS.

\emph{Notations}: Lowercase letter $x$, boldface lowercase letter $\mathbf{x}$, and boldface uppercase letter $\mathbf{X}$ are used to denote scalars, vectors, and matrices, respectively. The space of $N \times M$ real, complex, and binary matrices  are denoted by $\mathbb{R}^{N \times M}$, $\mathbb{C}^{N \times M}$, and $\mathbb{B}^{N \times M}$, respectively.  An $N\times N$ Hermitian matrix is represented by $\mathbb{H}^N$. $|\cdot|$, $\|\cdot\|$, and  $\|\cdot\|_{\mathrm{F}}$ denote the modulus of a complex-valued scalar, an Euclidean norm of a vector, and Frobenius norm of a matrix, respectively.  $\lambda_{\max}(\cdot)$, $\mathrm{det}(\cdot)$, $(\cdot)^{\mathrm{T}}$, $(\cdot)^{\mathrm{H}}$, $(\cdot)^{*}$, $\mathbb{E}\{\cdot\}$, $\mathrm{Rank}(\cdot)$, and $\mathrm{Tr(\cdot)}$ represent the largest eigenvalue, determinant, transpose, conjugate transpose, conjugate, expectation, rank, and trace of a matrix, respectively. $[x]^+ = \max\{0,x\}$. Positive semi-definite of matrix $\mathbf{X}$ is  represented by $\mathbf{X}\succeq\mathbf{0}$. $\mathrm{diag(\mathbf{x})}$ or $\widetilde{\mathrm{diag}}(\mathbf{X})$  denote a diagonal matrix whose entry is vector $\mathbf{x}$ or identical to that of $\mathbf{X}$, respectively.
$\mathcal{CN}(\mu, \sigma^2)$ denotes the distribution of a circularly symmetric complex Gaussian (CSCG) random variable with mean $\mu$ and variance $\sigma^2$, where $\sim$ represents ``distributed as''. $\mathbf{I}_N$ stands for an $N\times N$ identity matrix. 
\section{System Model}\vspace{-1mm}
\begin{figure}[t] \vspace*{-2mm}
  \centering
  \includegraphics[width=2.5in]{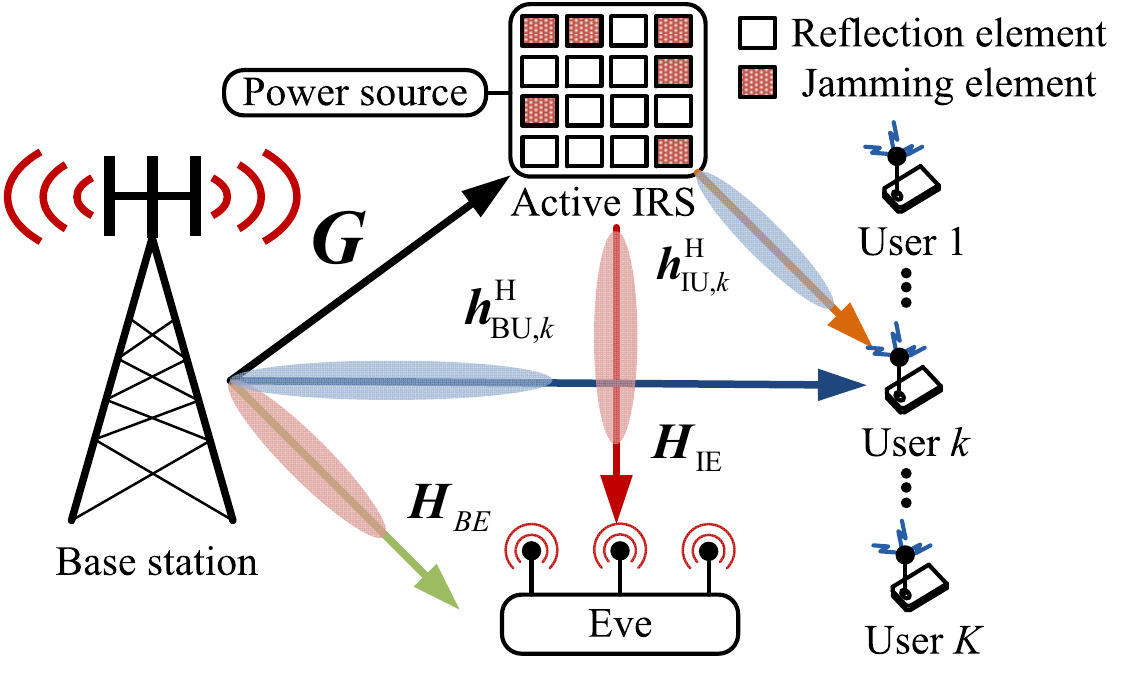}  \vspace*{-4mm}
  \caption{\hspace{-1mm}A secure wireless communication system with a multifunctional active IRS, a multi-antenna eavesdropper (Eve), and $K$ users.}
  \label{system_model} \vspace*{-7mm}
\end{figure}
The considered secure transmission system,  shown in Fig. \ref{system_model}, consists of four types of nodes, i.e., a base station (BS) with $N_{\mathrm{T}}$ antennas,  $K$ single-antenna users, an active IRS with $M$ elements, and an eavesdropper with $N_{\mathrm{E}}$ antennas\footnote{The consideration of an eavesdropper with  multiple antennas is equivalent to the case of cooperative multiple single-antenna eavesdroppers\cite{6470734}.}, where $ M+N_{\mathrm{T}}\geq N_{\mathrm{E}}$. Specifically, the BS intends to transmit confidential information to the $K$ users simultaneously. To enhance the communication security, both the BS and the active IRS are able to emit jamming signals to combat potential wiretapping.
A quasi-static flat fading channel model \cite{xu2021resource,9483903} is assumed and the channel state information (CSI) of all the channels is assumed to be perfectly known, which can be achieved by currently available channel estimation methods\footnote{To  obtain intuitive insights, we assume a perfect CSI model in this paper. The case of the imperfect CSI model will be considered in its journal version.} \cite{wang2020channel}. The baseband equivalent channels from the BS to the active IRS (BS-IRS), from the active IRS to  user $k$ (IRS-user$_k$), from the active IRS to the eavesdropper (IRS-Eve), from the BS to user $k$ (BS-user$_k$), and from the BS to the eavesdropper (BS-Eve) are denoted by   $\bm{\mathbf{G}}\in\mathbb{C}^{M \times N_{\mathrm{T}}}$, $\mathbf{h}_{\mathrm{IU},k}^{\mathrm{H}}\in\mathbb{C}^{1\times M}$, $\mathbf{H}_{\mathrm{IE}}\in\mathbb{C}^{N_{\mathrm{E}} \times M}$, $\mathbf{h}_{\mathrm{BU},k}^{\mathrm{H}}\in\mathbb{C}^{1\times N_{\mathrm{T}} }$, and $\mathbf{H}_{\mathrm{BE}}\in\mathbb{C}^{N_{\mathrm{E}} \times N_{\mathrm{T}}}$, respectively.
\subsection{Signal Model}
 \subsubsection{Transmitted Signal at the BS}
The transmitted signal from the BS is given by\vspace{-3mm}
\begin{align}
    \mathbf{x} = \sum_{k=1}^{K}\mathbf{w}_k x_k + \mathbf{z}_{\mathrm{B}},
    \label{tx_signal}\\[-8mm]\notag
\end{align}
where $\mathbf{w}_k\in\mathbb{C}^{N_{\mathrm{T}}\times1}, \forall k\in\{1,\ldots,K\}$, is the precoding vector for user $k$ and $x_k\sim \mathcal{CN}(0,1)$ is the data symbol intended for user $k$. The AN vector is represented by $\mathbf{z}_{\mathrm{B}}\in\mathbb{C}^{N_{\mathrm{t}}\times 1}$ which is generated by the BS to deliberately combat the eavesdropper. Specifically, $\mathbf{z}_{\mathrm{B}}$ is a CSCG vector with $\mathbf{z}_{\mathrm{B}}\sim\mathcal{CN}(\mathbf{0},\mathbf{Z}_{\mathrm{B}})$, where $\mathbf{Z}_{\mathrm{B}} = \mathbb{E}\{\mathbf{z}_{\mathrm{B}}\mathbf{z}_{\mathrm{B}}^{\mathrm{H}}\}$ is the covariance matrix of the AN vector. Notably, the AN is supposed to be unknown to both the legitimate users and the potential eavesdropper and the proposed algorithm will optimize the AN to facilitate communication security provisioning in the next section.
 \subsubsection{Reflected Signal at the IRS}\hspace{-1mm}
As shown in Fig. \ref{system_model},  some of the elements in the active IRS are selected for emitting jamming signals and the rest reflecting the impinging signals to the users. In particular, the reflection matrix of the IRS is denoted as $\mathbf{A}\bm{\Theta}$. Matrix $\bm{\Theta}= \mathrm{diag}(\bm{\phi}) \in \mathbb{C}^{M\times M}$, where $\bm{\phi} = [\phi_1,\ldots,\phi_m,\ldots,\phi_M]^{\mathrm{T}} \in \mathbb{C}^{M\times1}$ and $\phi_m= p_m e^{j\theta_m}$,  $\forall m \in [1,\ldots,M]$, is a diagonal matrix, which controls the phase shifts and amplitude gain introduced by the active IRS. Besides, $\theta_m \in [0,2\pi)$ and $p_m\geq 0$ are the phase shift and the amplitude coefficient for the $m$-th active IRS element, respectively. Thanks to the integrated active amplifier, $p_m$ can be lager than one as in contrast to its counterpart in passive IRSs, e.g. \cite{zhang2021active,xu2021resource,9377648}. Also, $\mathbf{A}= \mathrm{diag}(\bm{\alpha})\in \mathbb{B}^{M\times M}$ is a mode selection matrix that to be optimized, where $\bm{\alpha}=[\alpha_{1},\ldots,\alpha_{m},\ldots,\alpha_{M}]^{\mathrm{T}}\in \mathbb{B}^{M\times 1}$. Binary variable $\alpha_{m}\in\{0,1\}$, is an IRS mode selector, which is defined as:\vspace{-2mm}
\begin{align}
    \hspace{-3mm} \alpha_{m} \hspace{-1mm} = \hspace{-1mm} \left\{
    \begin{array}{ll}
      \hspace{-2mm}  1, &\hspace{-2mm} \text{Reflection mode at the active IRS element} \,m\text{,}\\
      \hspace{-2mm}  0, &\hspace{-2mm} \text{Jamming mode at the active IRS element} \,m\text{.}
    \end{array}
\right.\\[-7mm]\notag
\end{align}\begin{figure}[t] \vspace*{-3mm}
  \centering
  \includegraphics[width=2.7in]{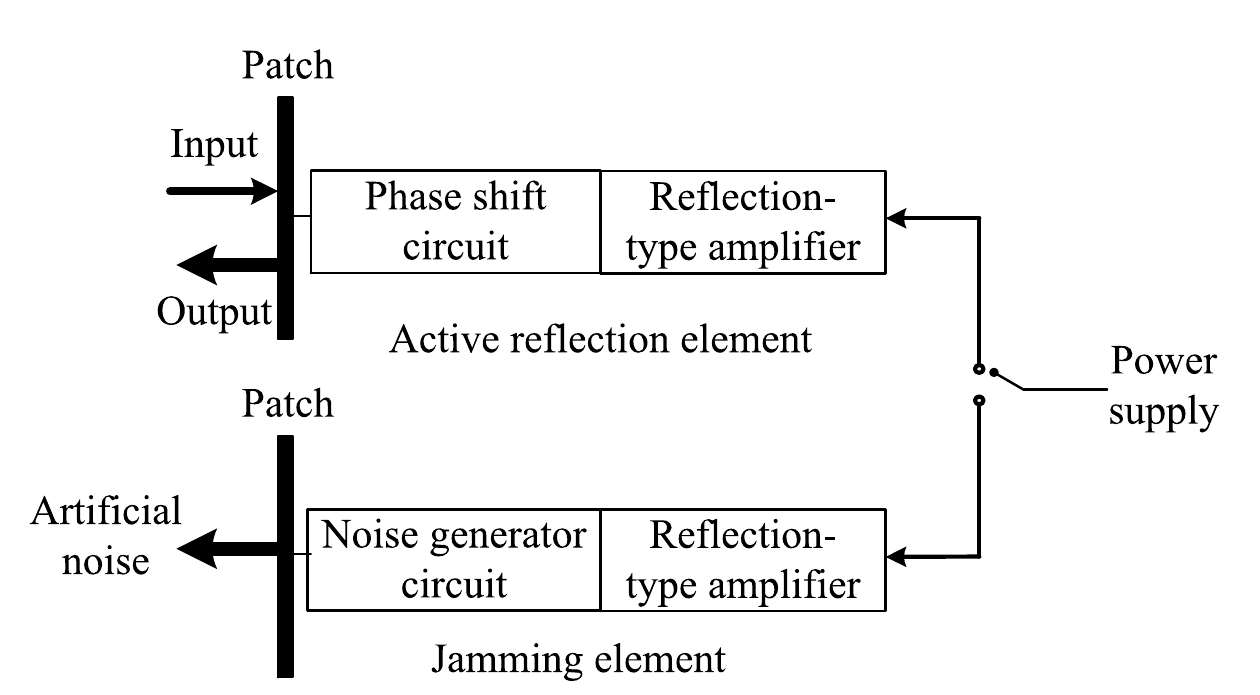}  \vspace*{-3mm}
  \caption{\hspace{-1mm}An architecture of the proposed active IRS\cite{zhang2021active,abdipour2008design}.}
  \label{element_circuits} \vspace*{-7mm}
\end{figure}
 As shown in Fig. \ref{element_circuits}, each of the active IRS elements can be set in either the reflection mode or the jamming mode. Specifically, active IRS elements in the reflection mode amplify and reflect all impinging signals, while the elements in the jamming mode generate a power-normalized AN and then amplify it to degrade the received signal quality of the eavesdropper.
 Moreover, the power budget of the active IRS is limited to $\|\bm{\Theta}\|^2_{\mathrm{F}}\leq P_{\mathrm{I}}^{\max}$, where $P_{\mathrm{I}}^{\max}$ is a constant.

Therefore, the overall signal  processing model at the proposed active IRS is given as \vspace{-2mm}
\begin{align}
\mathbf{y}_{\mathrm{I}} \hspace{-1mm} =  \hspace{-1mm}\underbrace{\mathbf{A}\mathbf{\Theta}\mathbf{G}\mathbf{x}}_{\text{Desired signal}}+\underbrace{(\mathbf{I}_{M}-\mathbf{A})\mathbf{\Theta}\mathbf{z}_{\mathrm{I}}}_{\text{AN}}+\underbrace{\bm{\Theta}\mathbf{n}_{\mathrm{I}}}_{\text{Dynamic noise}}+\underbrace{\cancel{\mathbf{n}_{\mathrm{s}}}}_{\text{Static noise}},\label{Eq:yIRS}\\[-7mm]\notag
\end{align}
where $\mathbf{z}_{\mathrm{I}} \in \mathbb{C}^{M\times1}$ denotes
the power-normalized AN  vector generated by the active IRS and $\mathbb{E}\{\mathbf{z}_{\mathrm{I}}\mathbf{z}_{\mathrm{I}}^{\mathrm{H}}\}=\mathbf{I}_{M}$.
As shown in \eqref{Eq:yIRS}, the power of the dynamic noise $\mathbf{\Theta}\mathbf{n}_{\mathrm{I}}$ is scaled with the power gain introduced by the active IRS, where $\mathbf{n}_{\mathrm{I}}$ is introduced by the input noise and inherent device noise. We model $\mathbf{n}_{\mathrm{I}}\in\mathbb{C}^{M\times1}$ as $\mathbf{n}_{\mathrm{I}}\sim \mathcal{CN}(\mathbf{0}_{M},\sigma_{\mathrm{I}}^2\mathbf{I}_M)$, where $\sigma_{\mathrm{I}}^2\geq 0$ is the corresponding power. Note that by optimizing $\mathbf{\Theta}$, the inherent noise generated by the active IRS, i.e., $\mathbf{\Theta}\mathbf{n}_{\mathrm{I}}$, can also serve as a jamming source to combat eavesdropping, rather than only being an interference to the system. On the other hand, different from the dynamic noise, the static noise $\mathbf{n}_{\mathrm{s}}$ is independent of amplification factors and it is usually negligible small compared with the dynamic noise\cite{bousquet20114}. As such, for simplicity, it is neglected in the sequel.
 \subsubsection{Signal Model of Users}
 The received signal at user $k$, $\forall k$, can be modeled as\vspace{-2mm}
 \begin{align}
    y_{\mathrm{U},k} = \mathbf{h}_{\mathrm{BU},k}^{\mathrm{H}}\mathbf{x} +  \mathbf{h}_{\mathrm{IU},k}^{\mathrm{H}}\mathbf{y}_{\mathrm{I}} +n_{\mathrm{U},k},\label{eq:y_U}\\[-7mm]\notag
\end{align}
where the thermal noise at user $k$ is represented by $n_{\mathrm{U},k}\sim \mathcal{CN}(0, \sigma^2_{\mathrm{U},k})$.
Therefore, the achievable rate   and the signal-to-interference-plus-noise ratio (SINR) at user $k$, $\forall k$, are \vspace{-0mm}
 \begin{align}
   R_{\mathrm{U},k} \hspace{-1mm}=&\log_2(1+\gamma_{\mathrm{U},k}) \text{ and}\label{eq:R_uk}\\[-1mm]
\gamma_{\mathrm{U},k} \hspace{-1mm}=& \frac{|\mathbf{h}_{\mathrm{eq},k}^{\mathrm{H}}\mathbf{w}_k|^2}
{\mu_k \hspace{-1mm}+\hspace{-1mm} |\mathbf{h}_{\mathrm{IU},k}^{\mathrm{H}}(\mathbf{I}_{M}\hspace{-1mm}-\hspace{-1mm}\mathbf{A})
\mathbf{\Theta}|^2\hspace{-1mm}+\hspace{-1mm}\sigma_{\mathrm{I}}^2\mathrm{Tr}(\mathbf{h}_{\mathrm{IU},k}^{\mathrm{H}}\bm{\Theta}\bm{\Theta}^{\mathrm{H}}\mathbf{h}_{\mathrm{IU},k})
},\notag\\[-7mm]\notag
\end{align}
respectively, where $\mu_k = \mathrm{Tr}(\mathbf{h}_{\mathrm{eq},k}^{\mathrm{H}}(\sum_{j\neq k}^{K}\mathbf{w}_j\mathbf{w}_j^{\mathrm{H}} + \mathbf{Z}_{\mathrm{B}})
\mathbf{h}_{\mathrm{eq},k})+\sigma^2_{\mathrm{U},k}$ and $\mathbf{h}_{\mathrm{eq},k}^{\mathrm{H}}=\mathbf{h}_{\mathrm{BU},k}^{\mathrm{H}}
+\mathbf{h}_{\mathrm{IU},k}^{\mathrm{H}}\mathbf{A}\mathbf{\Theta}\mathbf{G},\forall k$.

 \subsubsection{Signal Model of the Eavesdropper}
The received signal at the eavesdropper is given by\vspace{-2mm}
\begin{align}
    \mathbf{y}_{\mathrm{E}} =& \mathbf{H}_{\mathrm{BE}}\mathbf{x} +  \mathbf{H}_{\mathrm{IE}}\mathbf{y}_{\mathrm{I}}+\mathbf{n}_{\mathrm{E}}.\label{eq:y_E}\\[-7mm]\notag
\end{align}
Since the capability of the eavesdropper is unknown to the BS, for security provisioning, we assume that the eavesdropper has unlimited computational resources and is thus capable of cancelling all multiuser interference prior to decoding the information of a potential user. Besides, the thermal noise $\mathbf{n}_{\mathrm{E}}$ at the eavesdropper is ignored as it is also unknown to the BS. These worst-case assumptions lead to an upper bound on the channel capacity between the BS and the eavesdropper for decoding the signal of legitimate user $k$ \cite{zhou2021secure}, which is given\vspace{-2mm}
 \begin{align}
{C}_{\mathrm{E},k}\hspace{-1mm} =& \log_2\Big(\det\big(\mathbf{I}_{N_{\mathrm{E}}}+
     \mathbf{Q}^{-1}\mathbf{F}_{\mathrm{eq}}\mathbf{w}_k\mathbf{w}_k^{\mathrm{H}}\mathbf{F}_{\mathrm{eq}}^{\mathrm{H}}\big)\Big), \forall k,\label{eq:C_E}\\[-1mm]
     \mathbf{Q}\hspace{-1mm} =& \mathbf{F}_{\mathrm{eq}}\mathbf{Z}_{\mathrm{B}} \mathbf{F}_{\mathrm{eq}}^{\mathrm{H}} \hspace{-1mm}+\hspace{-1mm}\mathbf{H}_{\mathrm{IE}}(\mathbf{I}_{M}\hspace{-1mm}-\hspace{-1mm}\mathbf{A})\mathbf{\Theta}\mathbf{\Theta}^{\mathrm{H}}
     \mathbf{H}_{\mathrm{IE}}^{\mathrm{H}}\hspace{-1mm}+\hspace{-1mm}
     \sigma_{\mathrm{I}}^2\mathbf{H}_{\mathrm{IE}}\bm{\Theta}\bm{\Theta}^{\mathrm{H}}
     \mathbf{H}_{\mathrm{IE}}^{\mathrm{H}},\notag\\[-7mm]\notag
\end{align}
where $\mathbf{F}_{\mathrm{eq}}=\mathbf{H}_{\mathrm{BE}} +  \mathbf{H}_{\mathrm{IE}}\mathbf{A}\mathbf{\Theta}\mathbf{G}$.
 Hence, the achievable secrecy rate between  the BS and user $k$ is
 $R_{\mathrm{s},k} = [R_{\mathrm{U},k} - C_{\mathrm{E},k}]^+,\forall k$. 

\vspace{-1mm}
\section{Problem Formulation}\vspace{-1mm}

We aim to minimize the total system power consumption at the BS and the active IRS by jointly designing  $\mathbf{w}_k$, $\mathbf{Z}_{\mathrm{B}}$, $\alpha_m$, and $\mathbf{\Theta} $, which can be formulated as optimization problem $\mathcal{P}1$:\vspace{-2mm}
\begin{align}
\mathcal{P}1:\hspace{-3mm}&\underset{\mathbf{w}_k,\, \mathbf{Z}_{\mathrm{B}}\in\mathbb{H}^{N_{\mathrm{t}}},\,\alpha_m,\, \mathbf{\Theta}}{\mathrm{minimize}}\,\, \sum_{k=1}^{K}\|\mathbf{w}_k\|^2 +\mathrm{Tr}(\mathbf{Z}_{\mathrm{B}}) +\|\mathbf{\Theta}\|^2_{\mathrm{F}}  \label{proposed_formulation_origion} \\[-1mm]
&\mathrm{s.t.}\,\,\mathrm{C1}\hspace{-1mm}: \gamma_{\mathrm{U},k}  \geq \gamma_{k}^{\min}, \forall k, \notag\\[-1mm]
&\hspace{6mm}\mathrm{C2}\hspace{-1mm}: \log_2\Big(\det(\mathbf{I}_{N_{\mathrm{E}}}\!+\!
     \mathbf{Q}^{-1}\mathbf{F}_{\mathrm{eq}}\mathbf{w}_k\mathbf{w}_k^{\mathrm{H}}\mathbf{F}_{\mathrm{eq}}^{\mathrm{H}})\Big)\!\leq\! C_{k}^{\max}\!,\forall k, \notag\\[-1mm]
     &\hspace{6mm}\mathrm{C3}\hspace{-1mm}: \mathbf{Z}_{\mathrm{B}}\succeq\mathbf{\mathbf{0}},\hspace{0mm}\mathrm{C4}\hspace{-1mm}:  \alpha_{m} \in\{0,1\},\forall m,\mathrm{C5}\hspace{-1mm}: \|\mathbf{\Theta}\|^2\leq P_{\mathrm{I}}^{\max}.\notag\\[-7mm]\notag
\end{align}
Here, constraint $\mathrm{C1}$ ensures the minimum required SINR at user $k$ such that its achievable rate, i.e., $R_{\mathrm{U},k}$, exceeds $\log_2( 1+\gamma_{k}^{\min})$. In constraint $\mathrm{C2}$, the maximum tolerable information leakage to the eavesdropper for wiretapping the signals transmitted to user $k$ is limited to $C_{k}^{\max}$. By combining $\mathrm{C1}$ and $\mathrm{C2}$, the achievable secrecy rate between  the BS and user $k$ is bounded by below, i,e., $R_{\mathrm{s},k} \geq [R_{\mathrm{U},k} -C_{k}^{\max}]^{+}$, such that the minimum required  secrecy rate is guaranteed.
 $\mathbf{Z}_{\mathrm{B}}\in\mathbb{H}^{N_{\mathrm{t}}}$ and $\mathrm{C3}$ guarantees that  $\mathbf{Z}_{\mathrm{B}}$ satisfies the requirements of being a covariance matrix of the AN vector. Constraint $\mathrm{C4}$  requires that each active IRS element function in one of two modes: reflection mode or jamming mode. Constraint $\mathrm{C5}$ limits the power budget of the active IRS under $P_{\mathrm{I}}^{\max}$.
\section{Optimization Solution} \label{solution}
The formulated problem, $\mathcal{P}1$, is non-convex due to the binary variable $ \alpha_{m}$ in constraint $\mathrm{C4}$. Moreover, optimization variables are coupled in constraints $\mathrm{C1}$ and $\mathrm{C2}$ of the formulated problem $\mathcal{P}1$. Generally, obtaining the global optimization solution to $\mathcal{P}1$ requires an exhaustive search, which is computational expensive even for moderate system size. Therefore,  we present a computationally efficient iterative AO algorithm to obtain a sub-optimal solution to the problem $\mathcal{P}1$.
\subsection{Problem Transformation}\vspace{-0mm}
In this section, $\mathcal{P}1$ will be first rewritten as its equivalent form $\mathcal{P}2$, which paves the way to facilitate the development of the AO method. To start with, we introduce
the following proposition to handle the non-convexity of $\mathrm{C2}$.
\begin{Proposition}\label{Proposition_1} Since $\mathbf{Q}^{-\frac{1}{2}}\mathbf{F}_{\mathrm{eq}}\mathbf{w}_k\mathbf{w}_k^{\mathrm{H}}\mathbf{F}_{\mathrm{eq}}^{\mathrm{H}}\mathbf{Q}^{-\frac{1}{2}}$ is a rank-one matrix,
$\mathrm{C2}$ in $\mathcal{P}1$ can be equivalently transformed as\vspace{-2mm}
\begin{align}
\overline{\mathrm{C2}}\hspace{-1mm}:& C_{\mathrm{E},k}^{\mathrm{Tol}}\Big(\mathbf{F}_{\mathrm{eq}}\mathbf{Z}_{\mathrm{B}} \mathbf{F}_{\mathrm{eq}}^{\mathrm{H}} \hspace{-0mm}+\hspace{-0mm}\mathbf{H}_{\mathrm{IE}}(\mathbf{I}_{M}\hspace{-1mm}-\hspace{-1mm}\mathbf{A})\mathbf{\Theta}\mathbf{\Theta}^{\mathrm{H}}
     \mathbf{H}_{\mathrm{IE}}^{\mathrm{H}}\\[-2mm]
    & +\hspace{-0mm}
     \sigma_{\mathrm{I}}^2\mathbf{H}_{\mathrm{IE}}\bm{\Theta}\bm{\Theta}^{\mathrm{H}}
     \mathbf{H}_{\mathrm{IE}}^{\mathrm{H}}\Big) \hspace{-0mm}-\hspace{-0mm}\mathbf{F}_{\mathrm{eq}}\mathbf{w}_k\mathbf{w}_k^{\mathrm{H}}\mathbf{F}_{\mathrm{eq}}^{\mathrm{H}} \succeq\mathbf{0},\notag\\[-7mm]\notag
\end{align}\vspace{-2mm}
where $C_{\mathrm{E},k}^{\mathrm{Tol}}=2^{C_{k}^{\max}}-1$.
\end{Proposition}
\begin{proof}
Please refer to the appendix.
\end{proof}

Therefore, the $\mathcal{P}1$ can be rewritten as following:\vspace{-2mm}
\begin{align}
\mathcal{P}2:\underset{\mathbf{w}_k,\, \mathbf{Z}_{\mathrm{B}}\in\mathbb{H}^{N_{\mathrm{t}}},\,\alpha_m,\, \mathbf{\Theta}}{\mathrm{minimize}}\,\,& \sum_{k=1}^{K}\|\mathbf{w}_k\|^2 +\mathrm{Tr}(\mathbf{Z}_{\mathrm{B}}) +\|\mathbf{\Theta}\|^2_{\mathrm{F}}  \label{proposed_formulation_p2} \\[-1mm]
\mathrm{s.t.}\,\,&\mathrm{C1},\overline{\mathrm{C2}},\mathrm{C3}-\mathrm{C5}.\notag\\[-7mm \notag]
\end{align}
Since the variables in $\mathcal{P}2$ are still coupled, we divide $\mathcal{P}2$ into two sub-problems, which alternatingly solve $\{\mathbf{w}_k,\mathbf{Z}_{\mathrm{B}}\}$ and $\{\bm{\Theta},\alpha_m\}$ via fixing the other variables, respectively.

\subsection{Sub-problem 1: Optimization of Precoder and AN Vector at the BS}
Sub-problem 1 optimizes the  precoder vector, $\mathbf{w}_k$, and the AN covariance matrix,  $\mathbf{Z}_{\mathrm{B}}$, at the BS for a given phase shift matrix $\mathbf{\Theta}$ and mode selector $\alpha_m$. By defining $\mathbf{W}_k=\mathbf{w}_k\mathbf{w}_k^{\mathrm{H}}$, problem $\mathcal{P}2$ can be equivalently rewritten as:\vspace{-2mm}
\begin{align}
\mathcal{P}3:&\underset{\mathbf{W}_k\in\mathbb{H}^{N_{\mathrm{t}}},\, \mathbf{Z}_{\mathrm{B}}\in\mathbb{H}^{N_{\mathrm{t}}}}{\mathrm{minimize}}\,\, \sum_{k=1}^{K}\mathrm{Tr}(\mathbf{W}_k)+\mathrm{Tr}(\mathbf{Z}_{\mathrm{B}})   \label{proposed_formulation_p3} \\[-1mm]
&\mathrm{s.t.}\,\,
\overline{\mathrm{C1}}\hspace{-1mm}:\mathrm{Tr}(\mathbf{h}_{\mathrm{eq},k}^{\mathrm{H}}\mathbf{W}_k\mathbf{h}_{\mathrm{eq},k})\geq \gamma_{k}^{\min}\notag\\[-1mm]
&\hspace{5mm}\Big(\mathrm{Tr}(\mathbf{h}_{\mathrm{eq},k}^{\mathrm{H}}\mathbf{Z}_{\mathrm{B}}
\mathbf{h}_{\mathrm{eq},k})+\sum_{j\neq k}^{K}\mathrm{Tr}(\mathbf{h}_{\mathrm{eq},k}^{\mathrm{H}}\mathbf{W}_{j}
\mathbf{h}_{\mathrm{eq},k})\notag\\[-1mm]
&\hspace{6mm}+\mathrm{Tr}\big(\mathbf{h}_{\mathrm{IU},k}^{\mathrm{H}}(\mathbf{I}_{M}-\mathbf{A})
\mathbf{\Theta}\mathbf{\Theta}^{\mathrm{H}}\mathbf{h}_{\mathrm{IU},k}\big)\notag\\[-1mm]
&\hspace{6mm}+\sigma_{\mathrm{I}}^2\mathrm{Tr}(\mathbf{h}_{\mathrm{IU},k}^{\mathrm{H}}\bm{\Theta}\bm{\Theta}^{\mathrm{H}}\mathbf{h}_{\mathrm{IU},k})
+\sigma^2_{\mathrm{U},k}\Big), \forall k,\notag\\[-2mm]
&\overline{\mathrm{C2}}\hspace{-1mm}: C_{\mathrm{E},k}^{\mathrm{Tol}}\Big(\mathbf{F}_{\mathrm{eq}}\mathbf{Z}_{\mathrm{B}} \mathbf{F}_{\mathrm{eq}}^{\mathrm{H}} \hspace{-1mm}+\hspace{-1mm}\mathbf{H}_{\mathrm{IE}}(\mathbf{I}_{M}\hspace{-1mm}-\hspace{-1mm}\mathbf{A})\mathbf{\Theta}\mathbf{\Theta}^{\mathrm{H}}
     \mathbf{H}_{\mathrm{IE}}^{\mathrm{H}}\hspace{-1mm}\notag\\[-2mm]
&\hspace{6mm}+\hspace{-1mm}
     \sigma_{\mathrm{I}}^2\mathbf{H}_{\mathrm{IE}}\bm{\Theta}\bm{\Theta}^{\mathrm{H}}
     \mathbf{H}_{\mathrm{IE}}^{\mathrm{H}}\Big) \hspace{-1mm}-\hspace{-1mm}\mathbf{F}_{\mathrm{eq}}\mathbf{W}_k\mathbf{F}_{\mathrm{eq}}^{\mathrm{H}} \succeq\mathbf{0},\notag\\[-1mm]
&\mathrm{C3},\mathrm{C6}\hspace{-1mm}: \mathbf{W}_k \succeq \mathbf{0},\forall k,
{\mathrm{C7}\hspace{-1mm}: \mathrm{Rank}(\mathbf{W}_k )\leq 1,\forall k}.\notag\\[-7mm]\notag
\end{align}
Here, constraints $\mathrm{C6}$, $\mathrm{C7}$, and $\mathbf{W}_k\in\mathbb{H}^{N_{\mathrm{t}}}$ are imposed to guarantee that $\mathbf{W}_k=\mathbf{w}_k\mathbf{w}_k^{\mathrm{H}}$ still holds after optimizing $\mathbf{W}_k$. Note that rank-one constraint $\mathrm{C7}$ in \eqref{proposed_formulation_p3} is the only non-convex obstacle, which can be addressed by applying the semi-definite relaxation (SDR) technique \cite{9483903}, i.e., $\cancel{\mathrm{C7}\hspace{-1mm}: \mathrm{Rank}(\mathbf{W}_k)\leq 1,\forall k}$. Without constraint $\mathrm{C7}$, $\mathcal{P}3$ is a convex semidefinite programming that can be solved by some standard convex program solvers, e.g. CVX \cite{grant2014cvx} . In the following theorem, the tightness of the adopted SDR is studied.
\begin{theorem}
If $\mathcal{P}3$ is feasible, a rank-one solution of $\mathbf{W}_{k}$ in  \eqref{proposed_formulation_p3} can always be constructed.
\end{theorem}\vspace{-2mm}
\begin{proof}
Due to the page limit, we provide only a sketch of proof. By exploiting the Karush-Kuhn-Tucker (KKT) conditions of $\mathcal{P}3$, one can prove that there always exists a rank-one solution of $\mathbf{W}_k$ to ensure a bounded solution of its optimal dual problem of $\mathcal{P}3$. Furthermore, by manipulating the dual variables of its dual problem, we can construct the rank-one solution of  $\mathcal{P}3$.
\end{proof}\vspace{-2mm}
\subsection{Sub-problem 2: Optimization of Phase Shifts and the Mode Selection Matrix at the Active IRS}\vspace{-2mm}
In this section, we optimize phase shift matrix $\mathbf{\Theta}$ and mode selector $\alpha_m$ at the active IRS for given $\mathbf{W}_k$ and $\mathbf{Z}_{\mathrm{B}}$. Firstly, to address the binary constraint $\mathrm{C4}$ in $\mathcal{P}2$, we equivalently transform $\mathrm{C4}$ as \vspace{-0mm}%
\begin{align}
\mathrm{C4a}\hspace{-1mm}:\alpha_m -\alpha_m ^2\leq 0,\forall m,\,\,\,\, \mathrm{C4b}\hspace{-1mm}: 0\leq \alpha_m \leq 1,\forall m.
\end{align}
On the other hand, to decouple the coupling variables $\mathbf{A}\mathbf{\Theta}$, we introduce a slack optimization variable $\mathbf{u}=[u_1,\ldots,u_m,\ldots,u_M]^{\mathrm{T}}\in\mathbb{C}^{M\times1}$ and  $\mathrm{diag}(\mathbf{u}^{\mathrm{H}})=\mathbf{A}\mathbf{\Theta}$.  By adopt the Big-M formulation\cite{9483903}, $\mathrm{diag}(\mathbf{u}^{\mathrm{H}})=\mathbf{A}\mathbf{\Theta}$ can be converted into a set of equivalent constraints as follows:\vspace{-2mm}
\begin{align}
\mathrm{C8a}\hspace{-1mm}:& \,\phi_{m}^{*}-(1-\alpha_m)\frac{P_{\mathrm{I}}^{\max}}{M}\leq{u}_m,\forall m,
\\[-1mm]
\mathrm{C8b}\hspace{-1mm}:& \,{u}_m\leq \phi_{m}^{*}+(1-\alpha_m)\frac{P_{\mathrm{I}}^{\max}}{M},\forall m,\notag\\[-1mm]
\mathrm{C8c}\hspace{-1mm}:& -\alpha_m\frac{P_{\mathrm{I}}^{\max}}{M}\leq {u}_m,\forall m,
\mathrm{C8d}\hspace{-1mm}:{u}_m\leq \alpha_m\frac{P_{\mathrm{I}}^{\max}}{M},\forall m, \notag
\end{align}
 where $u_m$ is the $m$-th element of $\mathbf{u}$. When $\alpha_m = 0$, optimization variable $u_m$ is forced to zero by constraints $\mathrm{C8d}$ and $\mathrm{C8c}$. When $\alpha_m = 1$, while $u_m$ attains the same value as $\phi_{m}^{*}$ by constraints $\mathrm{C8a}$ and $\mathrm{C8b}$.
Therefore, constraints $\mathrm{C1}$ and $\mathrm{C2}$ are rewritten as $\overline{\overline{\mathrm{C1}}}$  and $\overline{\overline{\mathrm{C2}}}$, respectively:\vspace{-2mm}
\begin{align}
\overline{\overline{\mathrm{C1}}}\hspace{-1mm}:&\mathrm{Tr}\Big(\mathbf{h}_{\mathrm{eq},k}\mathbf{h}_{\mathrm{eq},k}^{\mathrm{H}}\big(\mathbf{W}_k-\gamma_{k}^{\min}(\sum_{j\neq k}^{K}\mathbf{W}_j + \mathbf{Z}_{\mathrm{B}})\big)
\Big)\notag\\[-2mm]
&\geq
 \gamma_{{k}}^{\min}\Big(
\mathrm{Tr}\big(\mathbf{h}_{\mathrm{IU},k}^{\mathrm{H}}
\mathbf{\Theta}\mathbf{\Theta}^{\mathrm{H}}\mathbf{h}_{\mathrm{IU},k}\big)
\notag\\[-1mm]
&-\mathrm{Tr}\big(\mathrm{diag}(\mathbf{h}_{\mathrm{IU},k}^{\mathrm{H}})\mathrm{diag}(\mathbf{h}_{\mathrm{IU},k})\mathbf{u}\mathbf{u}^{\mathrm{H}}\big)\notag\\[-1mm]
&+\sigma_{\mathrm{I}}^2\mathrm{Tr}(\mathbf{h}_{\mathrm{IU},k}^{\mathrm{H}}\bm{\Theta}\bm{\Theta}^{\mathrm{H}}\mathbf{h}_{\mathrm{IU},k})
+\sigma^2_{\mathrm{U},k}\Big), \forall k,\\[-1mm]
\overline{\overline{\mathrm{C2}}}\hspace{-1mm}: & C_{\mathrm{E},k}^{\mathrm{Tol}}\Big(\hspace{-0mm}(\sigma_{\mathrm{I}}^2+1)\mathbf{H}_{\mathrm{IE}}\mathbf{\Theta}\mathbf{\Theta}^{\mathrm{H}}
     \mathbf{H}_{\mathrm{IE}}^{\mathrm{H}}-\hspace{-0mm}\mathbf{H}_{\mathrm{IE}}\widetilde{\mathrm{diag}}(\mathbf{u}\mathbf{u}^{\mathrm{H}})
     \mathbf{H}_{\mathrm{IE}}^{\mathrm{H}}\hspace{-0mm}
     \Big) \hspace{-0mm}\notag\\[-1mm]
     &+\hspace{-0mm}\mathbf{H}_{\mathrm{BE}}\mathbf{B}_{k}\mathbf{H}_{\mathrm{BE}}^{\mathrm{H}} +\mathbf{H}_{\mathrm{BE}}\mathbf{B}_k\mathbf{G}^{\mathrm{H}}\mathrm{diag}(\mathbf{u})\mathbf{H}_{\mathrm{IE}}^{\mathrm{H}}
     \notag\\[-1mm]
     &+\mathbf{H}_{\mathrm{IE}}\mathrm{diag}(\mathbf{u}^{\mathrm{H}})\mathbf{G}\mathbf{B}_k\mathbf{H}_{\mathrm{BE}}^{\mathrm{H}}
     \notag\\[-1mm]
     &+\mathbf{H}_{\mathrm{IE}}\mathrm{diag}(\mathbf{u}^{\mathrm{H}})\mathbf{G}\mathbf{B}_k\mathbf{G}^{\mathrm{H}}
     \mathrm{diag}(\mathbf{u})\mathbf{H}_{\mathrm{IE}}^{\mathrm{H}}\succeq\mathbf{0}, \forall k, \label{C2_bar}\\[-7mm]\notag
\end{align}
respectively, where
$\mathbf{h}_{\mathrm{eq},k}\mathbf{h}_{\mathrm{eq},k}^{\mathrm{H}}=\mathbf{h}_{\mathrm{BU},k}\mathbf{h}_{\mathrm{BU},k}^{\mathrm{H}}
+\mathbf{h}_{\mathrm{BU},k}\mathbf{u}^{\mathrm{H}}\mathrm{diag}(\mathbf{h}_{\mathrm{IU},k}^{\mathrm{H}})\mathbf{G}
+\mathbf{G}^{\mathrm{H}}\mathrm{diag}(\mathbf{h}_{\mathrm{IU},k})\mathbf{u}\mathbf{h}_{\mathrm{BU},k}^{\mathrm{H}}
+\mathbf{G}^{\mathrm{H}}\mathrm{diag}(\mathbf{h}_{\mathrm{IU},k})\mathbf{u}\mathbf{u}^{\mathrm{H}}\mathrm{diag}(\mathbf{h}_{\mathrm{IU},k}^{\mathrm{H}})\mathbf{G}$, and
$\mathbf{B}_k= C_{k}^{\mathrm{Tol}}\mathbf{Z}_{\mathrm{B}}-\mathbf{W}_k.$ 
Now, we further simplify the last term of constraint $\overline{\overline{\mathrm{C2}}}$ in \eqref{C2_bar}.
 By performing the singular value decomposition (SVD),  $\mathbf{G}\mathbf{B}_k\mathbf{G}^{\mathrm{H}}=\bm{\Upsilon}_k\bm{\sum}_k\mathbf{V}_k^{\mathrm{H}}$, where $\bm{\Upsilon}_k=[\mathbf{r}_{k,1},\ldots,\mathbf{r}_{k,M}]$ and $\mathbf{V}_k=[\mathbf{v}_{k,1},\ldots,\mathbf{v}_{k,M}]$ are $M\times M$ complex unitary matrices, $\mathbf{r}_{k,m}\in\mathbb{C}^{M\times 1}$ and $\mathbf{v}_{k,m}\in\mathbb{C}^{M\times 1},\forall m$, are the orthonormal bases of $\bm{\Upsilon}_k$ and $\mathbf{V}_k$, respectively. $\bm{\sum}_k\in\mathbb{R}^{M\times M}$ is a diagonal matrix with entries $[\varrho_{k,1},\ldots,\varrho_{k,m}, \ldots,\varrho_{k,M}]$. Thus, SVD can be written as $\mathbf{G}\mathbf{B}_k\mathbf{G}^{\mathrm{H}}=\sum_{m=1}^{M}\mathbf{p}_{k,m}\mathbf{q}_{k,m}^{\mathrm{H}}$, where $\mathbf{p}_{k,m}=\varrho_{k,m}\mathbf{r}_{k,m}$ and $\mathbf{q}_{k,m}=\mathbf{v}_{k,m}, \forall m$.
 As such, the last term of constraint $\overline{\overline{\mathrm{C2}}}$ in \eqref{C2_bar} can be recast as\vspace{-1mm}
\begin{align}
&\mathbf{H}_{\mathrm{IE}}\mathrm{diag}(\mathbf{u}^{\mathrm{H}})\mathbf{G}\mathbf{B}_k\mathbf{G}^{\mathrm{H}}
     \mathrm{diag}(\mathbf{u})\mathbf{H}_{\mathrm{IE}}^{\mathrm{H}}\notag\\[-1mm]
        =&\sum_{m=1}^{M}\mathbf{H}_{\mathrm{IE}}\mathrm{diag}(\mathbf{p}_{k,m})\mathbf{u}\mathbf{u}^{\mathrm{H}}
     \mathrm{diag}(\mathbf{q}_{k,m}^{\mathrm{H}})\mathbf{H}_{\mathrm{IE}}^{\mathrm{H}}.\\[-7mm]\notag
\end{align}
Then, we tackle the remaining coupling matrices $\mathbf{\Theta}\mathbf{\Theta}^{\mathrm{H}}$ and $\mathbf{u}\mathbf{u}^{\mathrm{H}}$ in constraints $\overline{\overline{\mathrm{C1}}}$ and $\overline{\overline{\mathrm{C2}}}$. Firstly, we introduce slack optimization variables $\mathbf{\Phi}$ and $\mathbf{U}$, where\vspace{-1mm}
\begin{align}
\mathbf{\Phi}=\mathbf{\Theta}\mathbf{\Theta}^{\mathrm{H}}\text{ and }
\mathbf{U}=\mathbf{u}\mathbf{u}^{\mathrm{H}}.\label{eq:PhiTheta}\\[-7mm]\notag
\end{align}
  The two equalities in \eqref{eq:PhiTheta} can be equivalent written as the following constraints\vspace{-0mm}
\begin{align}
\mathrm{C9a}\hspace{-1mm}:& \begin{bmatrix}
\mathbf{\Phi}&\mathbf{\Theta}\\
\mathbf{\Theta}^{\mathrm{H}}&\mathbf{I}
\end{bmatrix}\succeq\mathbf{0},\hspace{8mm}
\mathrm{C9b}\hspace{-1mm}: \mathrm{Tr}(\mathbf{\Phi}-\mathbf{\Theta}^{\mathrm{H}}\mathbf{\Theta})\leq0,\\[-0mm]
\mathrm{C10a}\hspace{-1mm}:& \begin{bmatrix}
\mathbf{U}&\mathbf{u}\\
\mathbf{u}^{\mathrm{H}}&1
\end{bmatrix}\succeq\mathbf{0},\hspace{8mm}
\mathrm{C10b}\hspace{-1mm}: \mathrm{Tr}(\mathbf{U}-\mathbf{u}\mathbf{u}^{\mathrm{H}})\leq0.\notag\\[-7mm]\notag
\end{align}
Therefore, problem $\mathcal{P}3$ can be equivalently rewritten as\vspace{-1mm}
\begin{align}
\mathcal{P}4:&\underset{\mathbf{U}\in\mathbb{H}^{M},\mathbf{u},\mathbf{\Phi}\in\mathbb{H}^{M},\alpha_m,\,\mathbf{\Theta}}{\mathrm{minimize}}\,\,\mathrm{Tr}(\mathbf{\Phi}) \label{proposed_formulation_p4} \\[-1mm]
\mathrm{s.t.}\,\,&\mathrm{C4a},\mathrm{C4b},\mathrm{C8a}-\mathrm{C8d},\mathrm{C9a},
\mathrm{C9b},\mathrm{C10a},\mathrm{C10b}\notag\\[-1mm]
\overline{\overline{\mathrm{C1}}}\hspace{-1mm}:&\mathrm{Tr}\Big(\mathbf{h}_{\mathrm{eq},k}\mathbf{h}_{\mathrm{eq},k}^{\mathrm{H}}\big(\mathbf{W}_k-\gamma_{k}^{\min}(\sum_{j\neq k}^{K}\mathbf{W}_j + \mathbf{Z}_{\mathrm{B}})\big)
\Big)\notag
\\[-1mm]
&\geq
 \gamma_{k}^{\min}\Big(
\mathrm{Tr}\big(\mathbf{h}_{\mathrm{IU},k}^{\mathrm{H}}
\mathbf{\Phi}\mathbf{h}_{\mathrm{IU},k}\big)
-\mathrm{Tr}\big(\widetilde{\mathrm{diag}}(\mathbf{h}_{\mathrm{IU},k}^{\mathrm{H}}\mathbf{h}_{\mathrm{IU},k})\mathbf{U}\big)
\notag
\\&+\sigma_{\mathrm{I}}^2\mathrm{Tr}(\mathbf{h}_{\mathrm{IU},k}^{\mathrm{H}}\mathbf{\Phi}\mathbf{h}_{\mathrm{IU},k})
+\sigma^2_{\mathrm{U},k}\Big), \forall k,\notag\\[-0mm]
\overline{\overline{\mathrm{C2}}}\hspace{-1mm}: & C_{\mathrm{E},k}^{\mathrm{Tol}}\Big(\hspace{-0mm}(\sigma_{\mathrm{I}}^2+1)\mathbf{H}_{\mathrm{IE}}\mathbf{\Phi}
     \mathbf{H}_{\mathrm{IE}}^{\mathrm{H}}\hspace{-0mm}-\hspace{-0mm}\mathbf{H}_{\mathrm{IE}}\widetilde{\mathrm{diag}}(\mathbf{U})
     \mathbf{H}_{\mathrm{IE}}^{\mathrm{H}}\hspace{-0mm}
     \Big) \notag\\[-1mm]
     & \hspace{-0mm}+\hspace{-0mm}\mathbf{H}_{\mathrm{BE}}\mathbf{B}_{k}\mathbf{H}_{\mathrm{BE}}^{\mathrm{H}} +\mathbf{H}_{\mathrm{BE}}\mathbf{B}_k\mathbf{G}^{\mathrm{H}}\mathrm{diag}(\mathbf{u})\mathbf{H}_{\mathrm{IE}}^{\mathrm{H}}
     \notag\\[-0mm]
     &+\mathbf{H}_{\mathrm{IE}}\mathrm{diag}(\mathbf{u}^{\mathrm{H}})\mathbf{G}\mathbf{B}_k\mathbf{H}_{\mathrm{BE}}^{\mathrm{H}}
     \notag\\[-1mm]
     &+\sum_{m=1}^{M}\mathbf{H}_{\mathrm{IE}}\mathrm{diag}(\mathbf{p}_{k,m})\mathbf{U}
     \mathrm{diag}(\mathbf{q}_{k,m}^{\mathrm{H}})\mathbf{H}_{\mathrm{IE}}^{\mathrm{H}}\succeq\mathbf{0}.\notag
\notag\\[-0mm]
\mathrm{C5}\hspace{-1mm}:&\mathrm{Tr}(\mathbf{\Phi})\leq P_{\mathrm{I}}^{\max},\notag\\[-7mm]\notag
\end{align}
where $\mathbf{h}_{\mathrm{eq},k}\mathbf{h}_{\mathrm{eq},k}^{\mathrm{H}} = \mathbf{h}_{\mathrm{BU},k}\mathbf{h}_{\mathrm{BU},k}^{\mathrm{H}}+
\mathbf{h}_{\mathrm{BU},k}\mathbf{u}^{\mathrm{H}}\mathrm{diag}(\mathbf{h}_{\mathrm{IU},k})^{\mathrm{H}}\mathbf{G}
+\mathbf{G}^{\mathrm{H}}\mathrm{diag}(\mathbf{h}_{\mathrm{IU},k})\mathbf{u}\mathbf{h}_{\mathrm{BU},k}^{\mathrm{H}}
 +\mathbf{G}^{\mathrm{H}}\mathrm{diag}(\mathbf{h}_{\mathrm{IU},u})\mathbf{U}\mathrm{diag}(\mathbf{h}_{\mathrm{IU},k}^{\mathrm{H}})\mathbf{G} $.
\begin{table}[t] \vspace*{-6mm}
\scriptsize
\linespread{1.05}
\begin{algorithm} [H]
\caption{\small  SCA-based Iterative Active IRS Optimization} \label{alg_1}
\begin{algorithmic} [1]
\STATE Set the maximum iterations number $t_{\max}$, the index of the first iteration $t=0$, and optimization variables in $\mathbf{\Theta}^{(t)}$, $\mathbf{u}^{(t)}$, and $\alpha_m^{(t)},\forall m$.
\REPEAT[Main Loop: SCA]
\STATE Solve problem $\mathcal{P}5$ in \eqref{proposed_formulation_p5} with given optimization variables in $\mathbf{\Theta}^{(t)}$, $\mathbf{u}^{(t)}$, and $\alpha_m^{(t)}$, to obtain the variables for $\mathbf{\Theta}^{(t+1)}$, $\mathbf{u}^{(t+1)}$, and $\alpha_m^{(t+1)}$;
\STATE Set $t=t+1$ and update $\mathbf{\Theta}^{(t)}$, $\mathbf{u}^{(t)}$, and $\alpha_m^{(t)}$;
\UNTIL{convergence or $t=t_{\max}$}.\vspace*{-1mm}
\end{algorithmic}
\end{algorithm}
\vspace*{-10mm}
\end{table}
Now, the only non-convex constraints of problem $\mathcal{P}4$ are $\mathrm{C4a}$, $\mathrm{C9b}$, and $\mathrm{C10b}$, which are all in the form of difference of convex (d.c.) \cite{grant2014cvx,polik2010interior}, and therefore, we adopt the iterative successive convex approximation (SCA) algorithm to tackle their non-convexities. In particular, for any feasible points $\alpha_m^{(t)}$, $\mathbf{\Theta}^{(t)}$, and $\mathbf{u}^{(t)}$, where  $(t)$ is the iteration index for \textbf{Algorithm \ref{alg_1}}, the lower bounds of $\alpha_m ^2$, $\mathrm{Tr}(\mathbf{\Theta}^{\mathrm{H}}\mathbf{\Theta})$, and $\mathrm{Tr}(\mathbf{u}\mathbf{u}^{\mathrm{H}})$ can be derived by their
first-order Taylor expansions as\vspace{-2mm}
\begin{align}
\alpha_m ^2&\geq (\alpha_m ^{(t)})^2+2\alpha_m^{(t)}(\alpha_m -\alpha_m ^{(t)}),\forall m,\notag\\[-1mm]
\mathrm{Tr}(\mathbf{\Theta}^{\mathrm{H}}\mathbf{\Theta})&\geq -\|\mathbf{\Theta}^{(t)}\|_{\mathrm{F}}+2\mathrm{Tr}\Big((\mathbf{\Theta}^{(t)})^{\mathrm{H}}\mathbf{\Theta}\Big),\text{ and }\label{Eq:Taylor}\\[-2mm]
\mathrm{Tr}(\mathbf{u}\mathbf{u}^{\mathrm{H}})&\geq -\|\mathbf{u}^{(t)}\|+2\mathrm{Tr}\Big((\mathbf{u}^{(t)})^{\mathrm{H}}\mathbf{u}\Big),\notag\\[-7mm]\notag
\end{align}
respectively. Thus, by substituting the corresponding lower bounds in \eqref{Eq:Taylor} into $\mathrm{C4a}$, $\mathrm{C9b}$, and $\mathrm{C10b}$ in $\mathcal{P}4$, the subset of these constraints can be obtained as:\vspace{-2mm}
\begin{align}
\overline{\mathrm{C4a}}\hspace{-1mm}:&\alpha_m \leq (\alpha_m ^{(t)})^2+2\alpha_m ^{(t)}(\alpha_m -\alpha_m ^{(t)}),\forall m,\\[-1mm]
\overline{\mathrm{C9b}}\hspace{-1mm}: & \mathrm{Tr}(\mathbf{\Phi})\leq-\|\mathbf{\Theta}^{(t)}\|_{\mathrm{F}}+2\mathrm{Tr}\Big((\mathbf{\Theta}^{(t)})^{\mathrm{H}}\mathbf{\Theta}\Big),\notag\\[-1mm]
\overline{\mathrm{C10b}}\hspace{-1mm}:& \mathrm{Tr}(\mathbf{U})\leq-\|\mathbf{u}^{(t)}\|+2\mathrm{Tr}\Big((\mathbf{u}^{(t)})^{\mathrm{H}}\mathbf{u}\Big).\notag\\[-7mm]\notag
\end{align}
As such, the following optimization problem can yield a suboptimal solution of $\mathcal{P}4$.\vspace{-2mm}
\begin{align}
\mathcal{P}5:&\underset{\mathbf{U}\in\mathbb{H}^{M},\mathbf{u},\mathbf{\Phi}\in\mathbb{H}^{M},\alpha_m,\,\phi_m,\,\mathbf{\Theta}}{\mathrm{minimize}}\,\,\mathrm{Tr}(\mathbf{\Phi}) \label{proposed_formulation_p5} \\
\mathrm{s.t.}\,\,&\overline{\mathrm{C1}},\overline{\mathrm{C2}},\overline{\mathrm{C4a}},\!\mathrm{C4b},\mathrm{C5},\!\mathrm{C8a}\!-\!\mathrm{C8d},\!\mathrm{C9a},
\overline{\mathrm{C9b}},\mathrm{C10a},\overline{\mathrm{C10b}}.\notag\\[-7mm]\notag
\end{align}
\begin{table}[t] \vspace*{-2mm}
\scriptsize
\linespread{1.05}
\begin{algorithm} [H]
\caption{\small Overall Alternating Optimization Algorithm} \label{alg_overall}
\begin{algorithmic} [1]
\STATE Set the maximum iterations number $\tau_{\max}$, the index of the
first iteration $\tau=0$, and the optimization variables in $\mathbf{\Theta}^{(\tau)}$ and $\alpha_m^{(\tau)},\forall m$.
\REPEAT[Main Loop: Alternating Optimization]
\STATE Obtain $\mathbf{W}_k^{(\tau+1)}$ and $\mathbf{Z}_{\mathrm{B}}^{(\tau+1)}$ by solving problem \eqref{proposed_formulation_p3} with given optimization variables in $\mathbf{\Theta}^{(\tau)}$ and $\alpha_m^{(\tau)}$.
\STATE Obtain the variables $\mathbf{\Theta}^{(\tau+1)}$, $\mathbf{u}^{(\tau+1)}$, and $\alpha_m^{(\tau+1)}$ by \textbf{Algorithm \ref{alg_1}} with given $\mathbf{W}_k^{(\tau+1)}$ and $\mathbf{Z}_{\mathrm{B}}^{(\tau+1)}$;
\STATE Update $\mathbf{\Theta}^{(\tau)}=\mathbf{\Theta}^{(\tau+1)}$ and $\alpha_m^{(\tau)}=\alpha_m^{(\tau+1)}$.
\STATE Set $\tau=\tau+1$ and update the optimization variables;
\UNTIL{convergence or $\tau=\tau_{\max}$}.\vspace*{-1mm}
\end{algorithmic}
\end{algorithm}
\vspace*{-11mm}
\end{table}
This optimization problem is jointly convex with respect to $\mathbf{U},\mathbf{u},\mathbf{\Phi},\alpha_m$, and $\mathbf{\Theta}$, which can be solved by standard convex program solvers, e.g. CVX\cite{grant2014cvx}. The proposed algorithm for solving problem $\mathcal{P}5$ is shown in \textbf{Algorithm \ref{alg_1}}.
Moreover, \textbf{Algorithm \ref{alg_overall}} details the proposed overall AO algorithm, which solves  sub-problems $\mathcal{P}3$ and $\mathcal{P}5$ iteratively. With polynomial time computational complexity, \textbf{Algorithm \ref{alg_overall}}   is guaranteed to  converge to a suboptimal solution of  $\mathcal{P}1$ \cite{polik2010interior}.

\section{Results and Discussion}
This section examines the performance of the proposed scheme via simulations. The number of users is set to $K=2$. The location of the BS, the active IRS, users, and the eavesdropper are set in $(0,0)$, $(60, 20)$, $\{(100,10),(100,-10)\}$, and $(80,20)$ in a Cartesian coordinate system. We adopt the distance-dependent path loss model in \cite{wu2019intelligent} and the reference distance is set to $1$ m.
Unless otherwise indicated, other important parameters are listed as follows. The path loss exponents of the BS-user$_k$, BS-Eve, BS-IRS, IRS-user$_k$, and IRS-Eve links are set as $\eta_{\mathrm{BU,}k} = 3.5, \eta_{\mathrm{BE}} = 3.5$, $\eta_{\mathrm{BI}} = 2.6, \eta_{\mathrm{IU},k}=2.6$, and  $\eta_{\mathrm{IE}}=2.6$, respectively. Rician factors of the the BS-user$_k$, BS-Eve, BS-IRS, IRS-user$_k$, and IRS-Eve links are set as $\beta_{\mathrm{BU,}k} =0, \beta_{\mathrm{BE}} = 0$, $\beta_{\mathrm{BI}} =3,  \beta_{\mathrm{IU},k}= 3$, and $\beta_{\mathrm{IE}}=3$, respectively. The centre carrier frequency is $2.4$ GHz. The noise power at user $k$ and the active IRS are $\sigma^2_{\mathrm{U},k} = $ and $\sigma^2_{\mathrm{I}} = -100$ dBm. The number of the active IRS phase shifters, the BS's antennas, and the eavesdropper's antennas are $M=60$, $N_{\mathrm{T}} = 4$, and $N_{\mathrm{E}} = 2$, respectively. Without loss generality, we set the minimum requirement of SINR at all the users as the same value, i.e., $\gamma^{\min}_{k} = \gamma^{\min},\forall k$. The  maximum tolerable information leakage is set as $C_{k}^{\max} = 1.6$ bit/s/Hz for all users. The power consumption of the active IRS, $P^{\max}_{\mathrm{I}}=10$ dBm.
Moreover, we introduce two baseline schemes for system performance comparison: 1) Baseline scheme 1 is the same as the proposed scheme except that all the active IRS phase shifts are set to the reflection mode without emitting any AN; 2) Baseline scheme 2 does not deploy an IRS.

\begin{figure}[t]\vspace*{-2mm}
  \centering
  \includegraphics[width=3.3in]{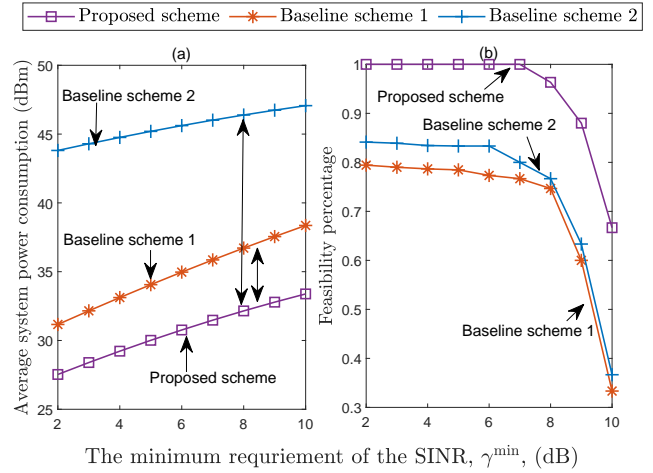}  \vspace*{-3mm}
  \caption{\hspace{-1mm}Average system power consumption versus the minimum requirement of SINR at users $ \gamma^{\min}$ with $N_{\mathrm{E}}=2$ and $M=60$.}
  \label{diffRmin} \vspace*{-8mm}
\end{figure}

Fig. \ref{diffRmin} (a) demonstrates the average system power consumption versus the minimum requirement of SINR at users, $ \gamma^{\min}$, for different schemes. In particular, Fig. \ref{diffRmin} (a) only records the results of the realizations that all the considered schemes have obtained a feasible solution to \eqref{proposed_formulation_origion}. For a fair comparison, we also show the feasibility percentage for solving \eqref{proposed_formulation_origion} versus  $ \gamma^{\min}$ for all the schemes in Fig. \ref{diffRmin} (b). As shown in Fig. \ref{diffRmin} (a), the average system power consumption increases with $\gamma^{\min}$ for all the considered schemes, since a more powerful beamforming is needed to satisfy the more stringent requirements of SINR.
Moreover, as active IRS is not available for baseline scheme 2, it is less capable of shaping the beamformer and the AN to mitigate the potential information leakage. This leads to higher power consumption to satisfy the quality of service (QoS) requirement compared with baseline scheme 1 and the proposed scheme.
In comparison to baseline scheme 2, the proposed scheme enjoys a substantial power reduction. This is because the active IRS elements in the proposed scheme can be optimized to reflect the incident signals or emit AN adapting to the channel conditions, which provides an additional DoF to combat the potential wiretapping efficiently. Note that the highest flexibility of the proposed scheme in exploiting the DoF also leads to the highest feasibility percentage over other considered baseline schemes, as shown in Fig. \ref{diffRmin} (b).

\begin{figure}[t]
  \centering
  \includegraphics[width=2.8in]{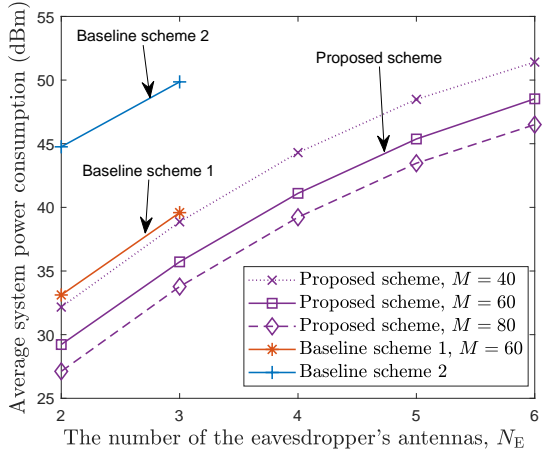}  \vspace*{-3mm}
  \caption{\hspace{-1mm}Average system power consumption versus the number of the eavesdropper's antennas with $  \gamma^{\min}=4$ dB. Baseline schemes 1 and 2 are infeasible in all the channel realizations for $N_{\mathrm{T}}\leq N_{\mathrm{E}}$, while that for the proposed scheme are feasible.}
  \label{diffN} \vspace*{-6mm}
\end{figure}

Fig. \ref{diffN} demonstrates the average system power consumption versus the number of the eavesdropper's antennas for cases with a different number of elements at the active IRS.
Since there is insufficient DoF and flexibility for optimizing the system resources, baseline schemes 1 and 2 are infeasible in all channel realizations when $N_{\mathrm{E}} \geq N_{\mathrm{T}}$. In contrast, the proposed scheme admits a feasible solution in the same setting as the former. In particular, as  $N_{\mathrm{E}}$ increases, the eavesdropper enjoys more spatial DoF for more effective wiretapping. As such, a higher power is allocated to AN at both the BS and the active IRS to neutralize the decoding capability of the eavesdropper and thus the total power consumption increases in the proposed scheme.
Furthermore, the average system power consumption can be reduced by increasing the number of active IRS elements. Indeed, the extra spatial DoF offered by the additional IRS elements provides higher flexibility to enhance the signal quality between the BS and the users while effectively mitigating the information leakage to the eavesdropper. 

\vspace{-1mm}
\section{Conclusions}\vspace{-1mm}
We investigated a multi-user MISO downlink communications system with an active IRS in the presence of a multi-antenna eavesdropper. Specifically, both the BS and the active IRS are able to emit the AN for security provisioning. A power minimization design was formulated as a non-convex optimization problem while guaranteeing the secrecy rate between the BS and the users. Our design optimized the precoder and the AN vector at the BS and the amplitude, phase, and mode selection of each active IRS element. A suboptimal solution to the design problem was obtained by applying AO, SCA, and SDR techniques. Simulation results unveiled that the proposed scheme can provide a significant power consumption reduction than the other considered baseline schemes to guarantee the same QoSs. Furthermore, even when the number of eavesdropper antennas exceeds that of the BS, the proposed scheme can still provide secure communication.\vspace{-1mm}
\appendix\label{app_prof_p1}\vspace{-1mm}
According to the Weinstein–Aronszajn identity\cite{sylvester1851xxxvii}, $\mathrm{det}(\mathbf{I}+\mathbf{AB})=\mathrm{det}(\mathbf{I}+\mathbf{BA})$, and thus ${C}_{\mathrm{E},k}$ in constraint $\mathrm{C2}$ of $\mathcal{P}1$ is equivalent as\vspace{-2mm}
 \begin{align}
& \hspace{-1mm}\log_2\Big(\det(\mathbf{I}_{N_{\mathrm{E}}}\hspace{-1mm}+\hspace{-1mm}
     \mathbf{Q}^{-1}\mathbf{F}_{\mathrm{eq}}\mathbf{w}_k\mathbf{w}_k^{\mathrm{H}}\mathbf{F}_{\mathrm{eq}}^{\mathrm{H}})\Big)\leq C_{k}^{\max}\notag\\[-1mm]
       \Leftrightarrow  \,\,& \log_2\big(1+\hspace{-1mm}
    \mathrm{Tr}( \mathbf{w}_k^{\mathrm{H}}\mathbf{F}_{\mathrm{eq}}^{\mathrm{H}} \mathbf{Q}^{-1}\mathbf{F}_{\mathrm{eq}}\mathbf{w}_k)\big)\leq {C_{k}^{\max}}
       \notag\\[-1mm]
       \Leftrightarrow  \,\,& \mathbf{w}_k^{\mathrm{H}}\mathbf{F}_{\mathrm{eq}}^{\mathrm{H}}\mathbf{Q}^{-1}
       \mathbf{F}_{\mathrm{eq}}\mathbf{w}_k\leq C_{\mathrm{E},k}^{\mathrm{Tol}}\notag
       \\[-1mm]
       \Leftrightarrow  \,\,&
 \mathrm{Tr}(\mathbf{Q}^{-1}\mathbf{F}_{\mathrm{eq}}\mathbf{w}_k\mathbf{w}_k^{\mathrm{H}}
 \mathbf{F}_{\mathrm{eq}}^{\mathrm{H}})\leq C_{\mathrm{E},k}^{\mathrm{Tol}}\notag\\[-1mm]
    \stackrel{\text{(a)}}{\Leftrightarrow}   \,\,&
     \lambda_{\max}\Big\{\mathbf{Q}^{-\frac{1}{2}}\mathbf{F}_{\mathrm{eq}}\mathbf{w}_k\mathbf{w}_k^{\mathrm{H}}\mathbf{F}_{\mathrm{eq}}^{\mathrm{H}}\mathbf{Q}^{-\frac{1}{2}}\Big\}\leq C_{\mathrm{E},k}^{\mathrm{Tol}}\notag\\[-1mm]
     \Leftrightarrow \,\,& C_{\mathrm{E},k}^{\mathrm{Tol}}\mathbf{Q} -\mathbf{F}_{\mathrm{eq}}\mathbf{w}_k\mathbf{w}_k^{\mathrm{H}}\mathbf{F}_{\mathrm{eq}}^{\mathrm{H}} \succeq\mathbf{0}.\label{eq:lemma_1_proof}\\[-7mm]\notag
\end{align}
Note that as $\mathbf{Q}^{-\frac{1}{2}}\mathbf{F}_{\mathrm{eq}}\mathbf{w}_k\times\mathbf{w}_k^{\mathrm{H}}
\mathbf{F}_{\mathrm{eq}}^{\mathrm{H}}\mathbf{Q}^{-\frac{1}{2}}$ is a rank one matrix, (a) holds.
Then, substitute $\mathbf{Q}$ in \eqref{eq:C_E} into \eqref{eq:lemma_1_proof}, the result follows immediately.

\bibliographystyle{IEEEtran}
\bibliography{ActiveIRS-Eve}

\end{document}